\documentclass[10pt,twocolumn,aps,pra,amsmath,amssymb,showpacs]{revtex4-2}
\usepackage{bm}
\usepackage{mathrsfs}
\usepackage{graphicx}
\usepackage[usenames,dvipsnames,svgnames,table]{xcolor}
\usepackage[unicode=true,
            pdfusetitle, 
            bookmarks=true,
            bookmarksnumbered=false,
            bookmarksopen=false,
            breaklinks=true,
            pdfborder={0 0 0},
            backref=false,
            colorlinks=true]{hyperref}
\hypersetup{linkcolor=NavyBlue,urlcolor=NavyBlue,citecolor=NavyBlue}

\usepackage{amsthm}
\newtheorem{prop}{\protect\propositionname}
\providecommand{\propositionname}{Proposition}

\newtheorem{lemma}{Lemma}
\newtheorem{cor}{Corollary}

\usepackage{tikz-cd}

\newcommand{\cH}{\mathcal{H}}

\begin{document}

\title{Deviation from linear reduced dynamics always occurs for each non-factorisable system-environment unitary evolution}
\author{Iman Sargolzahi}
\email{sargolzahi@um.ac.ir}
\affiliation{Department of Physics, Faculty of Science, Ferdowsi University of Mashhad, Mashhad, Iran}

\begin{abstract}
In the simplest approximation, the reduced dynamics of a quantum system $S$ interacting with its environment $E$ is considered to be given by a completely positive   map. But, in general, this is not the case. In fact, the reduced dynamics of the system in not even linear, in general. Whether the reduced dynamics is linear or not is determined by two factors: the set of possible initial states of the system-environment $\mathcal{S}=\left\lbrace \rho_{SE} \right\rbrace $, and the joint system-environment unitary evolution $U$. When $U$ is factorisable as $U=U_S\otimes U_E$, then we can choose $\mathcal{S}=\mathcal{D}$, where $\mathcal{D}$ is the set of all system-environment density operators. In other words, when $U$ is factorisable, the reduced dynamics of the system $S$ is linear (in fact unitary)  for arbitrary initial state of the system-environment $\rho_{SE}$. We show that this result cannot be generalized to any  non-factorisable $U$: For any non-factorisable unitary evolution of the whole system-environment $U$, the set $\mathcal{S}$ must be chosen as a proper subset of $\mathcal{D}$ to achieve linear reduced dynamics.
As a byproduct, considering a convex set of possible initial states of the system-environment $\mathcal{S}$ such that $\mathrm{Tr}_{E} \ \mathcal{S}=\mathrm{Tr}_{E} \ \mathcal{D}$, 
 we show that when the reduced dynamics of the system   for one system-environment unitary evolution $U_1$   is positive, but not completely positive,    this  implies that the reduced dynamics is not linear for another $U_2$.
 \end{abstract}


\maketitle

\section{Introduction} \label{sec:A}
The reduced dynamics of a quantum system $S$ interacting with its environment $E$   is given by
\begin{equation}
\label{eq:1a}
\rho_{S}^{\prime}=\mathrm{Tr}_{E} \circ \mathrm{Ad}_U (\rho_{SE}) \equiv \mathrm{Tr}_{E}(U \rho_{SE} U^{\dagger}),
\end{equation} 
where $\rho_{SE}$ is the initial state of the system-environment, $\rho_{S}^{\prime}$ is the final state of the system, and $U$ is a unitary operator on the whole system-environment \cite{1a}.

In the most famous case, one assumes that the set of possible initial states of the system-environment is factorized as  $\mathcal{S}=\lbrace \rho_{S}\otimes \tilde{\omega}_E \rbrace$, where $\rho_{S}$ are  arbitrary initial states of the system, but $ \tilde{\omega}_E $ is a fixed initial state of the environment. Then, the reduced dynamics of the system in Eq.  (\ref{eq:1a}) is given by a completely positive trace-preserving (CP) map, i.e., $\rho_{S}^{\prime}= \mathcal{E}_S (\rho_{S})$, where $\mathcal{E}_S$ is a CP map  \cite{1a}.

Assuming the set of possible initial states of the system-environment  as  $\mathcal{S}=\lbrace \rho_{S}\otimes \tilde{\omega}_E \rbrace$ seems to be an experimentally relevant approximation: The experimenter, who has access only to the system $S$ and not to the environment $E$, prepares different initial states of the system $\rho_{S}$, and then lets the system  to interact with the environment whose fixed (maybe unknown) initial state is $\tilde{\omega}_E$.
Therefore, this assumption  has been widely used in the theory of open quantum systems  \cite{2a, 3a, 4a}.

 Recently, some works have focused on relaxing this assumption and studied the consequences of considering more general sets $\mathcal{S}$ as the set of possible initial states of the system-environment.
The first consequence of such studies is that it has been  made clear that the CP-ness of the reduced dynamics is not a general fundamental result, as was considered previously, but it is only the consequence of considering $\mathcal{S}$ as $\lbrace \rho_{S}\otimes \tilde{\omega}_E \rbrace$, which is though an experimentally  relevant approximation, but it is not the general possible   $\mathcal{S}$  \cite{5a, 6a, 7a}.

How is the reduced dynamics for a general
  set $\mathcal{S}=\left\lbrace \rho_{SE} \right\rbrace $  of possible initial states of the system-environment?
 In fact, the final state of the system $\rho_{S}^{\prime}$ in Eq.  (\ref{eq:1a}) cannot be written as a function of its initial state $\rho_{S}=\mathrm{Tr}_{E}(\rho_{SE})$, in general  \cite{8a, 7a}.
Even if $\rho_{S}^{\prime}$ can be written as a function of $\rho_{S}$, this function is not linear  in general  \cite{9a, 10a}. And if, for each  $\rho_{S} \in  \mathcal{S}_S \equiv \mathrm{Tr}_{E} \ \mathcal{S}$, we have
  $\rho_{S}^{\prime}=\Phi_S(\rho_{S})$, where $\Phi_S$ is a linear map, this $\Phi_S$ is not CP in general, but it is  a (trace-preserving) Hermitian map, i.e., it maps each Hermitian operator to a Hermitian one \cite{11a}.
Obviously, when $\Phi_S$ is a non-positive map, then  $\mathcal{S}_S$ is a proper subset of $\mathcal{D}_S$,   where $\mathcal{D}_S$ is the set of all system  density operators \cite{11b}.

The next important class of results in this context is  finding other sets  $\mathcal{S}$ of possible initial states of the system-environment,  than the factorized one $\lbrace \rho_{S}\otimes \tilde{\omega}_E \rbrace$, 
     which they  also lead to CP redued dynamics \cite{12a, 13a, 14a, 15a}. All these sets  $\mathcal{S}$ introduced in Refs. \cite{12a, 13a, 14a, 15a}  are in fact   special cases of the most general possible one introduced in Refs.  \cite{16a, 17a} (see also Ref. \cite{18a}).

Among  other studies on correlated initial states of the system-environment, we can mention the effect of such correlation on process  
tomography \cite{19a, 20a, 21a}, detecting these correlations by monitoring only the reduced dynamics of the system  \cite{21b, 21c, 21a}, introducing a general framework for linear reduced dynamics \cite{7a,  24b},
 efforts to
 introduce other schemes to study the reduced dynamics \cite{19a,  22a}, unexpected entanglement dynamics  \cite{23a, 24a, 25a, 26a}, and the
 master equations and other tools of studding open quantum systems in the presence of initial correlation \cite{30a, 31a, 32a}.

From Eq. (\ref{eq:1a}), it is clear that the reduced dynamics of the system, in addition to the set of possible initial states $\rho_{SE}$, depends on the unitary system-environment evolution $U$ too.
When $U$ is factorized as $U_S\otimes U_E$, where  $U_S$ and $U_E$ are unitary operators on the system $S$ and the environment $E$, respectively, then, trivialy, the reduced dynamics of the system is CP (in fact  unitary) for arbitrary initial state of the system-environment  $\rho_{SE}$ \cite{33a}. In other words, when $U=U_S\otimes U_E$, choosing  $\mathcal{S}=\mathcal{D}$, where $\mathcal{D}$ is the set of all system-environment density operators, results in the CP-ness of  the reduced dynamics of the system $S$.

Does a similar result hold for a non-factorisable $U$ too? In other words, whether it is   possible to find a non-factorisable $U$ such that the reduced dynamics of the system $S$ is CP for arbitrary system-environment initial state  $\rho_{SE} \in\mathcal{D} $. This question has been addressed in Refs. \cite{7a, 34a}.

The case that the system-environment evolution is governed by a time-independent Hamiltonian has been considered 
in Ref.  \cite{34a}.  Decompose the total Hamiltonian as
\begin{equation}
\label{eq:2a}
H=H_S + H_E + H_{int},
\end{equation} 
where $H_S$ and  $H_E$ act on the system and the environment, respectively, and  $H_{int}$ is the interaction term between the system and the environment. In Ref. \cite{34a}, it has been shown that the reduced dynamics of the system $S$ is CP for arbitrary system-environment initial state  $\rho_{SE} \in\mathcal{D} $ if and only if 
$H_{int}=0$. In other words, the reduced dynamics is CP for arbitrary  $\rho_{SE}$ if and only if 
\begin{equation}
\label{eq:3a}
U=e^{-iHt}=e^{-iH_St}\otimes e^{-iH_Et}=U_S\otimes U_E.
\end{equation} 

In this paper, we generalize the above result to include arbitrary system-environment  $U$, not only those   which can be written as $U=e^{-iHt}$, and to include arbitrary linear reduced dynamics, not only the CP ones. In other words, we show that  the reduced dynamics of the system $S$ is linear for arbitrary system-environment initial state  $\rho_{SE} \in\mathcal{D} $ if and only if the whole system-environment unitary evolution  is factorisable as $U=U_S\otimes U_E$.

This result is given in sections ~\ref{sec:C} and ~\ref{sec:D}. Before, in Sec. ~\ref{sec:B}, we prove a Lemma needed to achieve the results of sections ~\ref{sec:C} and ~\ref{sec:D}. 

We give our next result in Sec.~\ref{sec:E}: We consider the case that  the set of possible initial states of the system-environment $\mathcal{S}$ is convex and $\mathrm{Tr}_{E} \ \mathcal{S}= \mathcal{D}_S$. We show that if the reduced dynamics of the system $S$, for one system-environment unitary evolution $U_1$, is positive but non-CP, then the reduced dynamics  is surely  non-linear for another $U_2$

 Finally, we end our paper in Sec.~\ref{sec:F}  with a summary of our results.

\section{$U$-invariance of the kernel of the partial trace over the environment}\label{sec:B} 

Consider the set $\mathcal{S}=\left\lbrace \rho_{SE} \right\rbrace $  of possible initial states of the system-environment (in an experiment). Each $\rho_{SE} \in \mathcal{S}$ can be written as 
\begin{equation}
\label{eq:4a}
\rho_{SE}= \rho_{S}\otimes\tilde{\omega}_E + R(\rho_{SE}),
\end{equation} 
where $\rho_{S}=\mathrm{Tr}_{E}(\rho_{SE})$ is the initial state of the system,  $ \tilde{\omega}_E $ is some fixed  state of the environment and $R(\rho_{SE})$ is a Hermitian operator such that $\mathrm{Tr}_{E} (R)=0$. In   Eq. (\ref{eq:4a}), we explicitly denote the dependence of $R$ on $\rho_{SE}$.

So, from Eq. (\ref{eq:1a}), the final state of the system is given by

\begin{equation}
\label{eq:5a}
\rho_{S}^\prime=\mathcal{E}_S (\rho_{S})+ \mathrm{Tr}_{E}(R^\prime),
\end{equation}
where $\mathcal{E}_S (*) =\mathrm{Tr}_{E} \circ \mathrm{Ad}_U (* \otimes \tilde{\omega}_E)$ is a CP map and $R^\prime = \mathrm{Ad}_U (R)$. Therefore, for a given $U$, if  $\mathrm{Tr}_{E}(R^\prime)=0$ for all $\rho_{SE} \in \mathcal{S}$, then the reduced dynamics is given by the CP map $\mathcal{E}_S$. 

Now, we can restate our   question as the following: Assuming that $\mathcal{S}=\mathcal{D}$ and $U$ is non-factorisable, whether $\mathrm{Tr}_{E}(R^\prime)=0$ for all $\rho_{SE} \in \mathcal{S}$ or not.

The kernel of a linear map  is the set of all inputs which are mapped to zero by this map. So, the kernel of the partial trace over the environment  is
\begin{equation}
\label{eq:6a}
\mathrm{ker} \ \mathrm{Tr}_E = \left\lbrace Y \in \mathcal{L}(\cH_S \otimes \cH_E) \mid  \mathrm{Tr}_E(Y)=0 \right\rbrace, 
\end{equation} 
where $\cH_S$ and $\cH_E$ are finite-dimensional Hilbert spaces of the system and the environment, respectively, and $\mathcal{L}(\cH)$ is the vector space of linear operators on the Hilbert space $\cH$.

Let's denote the set $\mathcal{R}$ as the set of all Hermitian operators $R$ in Eq. (\ref{eq:4a}), for all  $\rho_{SE} \in \mathcal{S}$. Obviously, $\mathcal{R}$ is a proper subset of $\mathrm{ker} \ \mathrm{Tr}_E$, and our  question is whether the set $\mathrm{Ad}_U \ \mathcal{R}$ is also a subset of $\mathrm{ker} \ \mathrm{Tr}_E$.
In this section, we focus on a more general question: Whether the set $\mathrm{Ad}_U \ \mathrm{ker} \ \mathrm{Tr}_E$ is a  subset of $\mathrm{ker} \ \mathrm{Tr}_E$, and we return to our  question in the next section.

\begin{lemma}
\label{lem1}
 $\mathrm{Tr}_{E} \circ \mathrm{Ad}_U (Y)=0$ for all  $Y \in \mathrm{ker} \ \mathrm{Tr}_E$, if and only if $U$ is factorisable as $U_S \otimes U_E$.
 \end{lemma}
\textit{Proof.}
When $U$ is factorisable, then, trivially, $\mathrm{Tr}_{E} \circ \mathrm{Ad}_U (Y)=0$ for each  $Y \in \mathrm{ker} \ \mathrm{Tr}_E$. Therefore, in the following, we focus on the   inverse.

We choose a $Y \in \mathrm{ker} \ \mathrm{Tr}_E$ as  
\begin{equation}
\label{eq:7a}
Y=A_S\otimes ( \vert \psi_{1E}\rangle\langle \psi_{1E}\vert - \vert \psi_{2E}\rangle\langle\psi_{2E}\vert), 
\end{equation}
where $A_S \in \mathcal{L}(\cH_S)$ and $ \vert \psi_{1E}\rangle, \  \vert \psi_{2E}\rangle \in \cH_E$ are normalized kets. So, the condition  $\mathrm{Tr}_{E} \circ \mathrm{Ad}_U (Y)=0$ reads
\begin{equation}
\label{eq:8a}
\mathcal{E}_{1S}(A_S)=\mathcal{E}_{2S}(A_S), 
\end{equation} 
where $\mathcal{E}_{1S} (*) =\mathrm{Tr}_{E} \circ \mathrm{Ad}_U (* \otimes  \vert \psi_{1E}\rangle\langle \psi_{1E}\vert)$ and $\mathcal{E}_{2S} (*) =\mathrm{Tr}_{E} \circ \mathrm{Ad}_U (* \otimes  \vert \psi_{2E}\rangle\langle \psi_{2E}\vert)$ are CP maps. Since $A_S$ in Eq. (\ref{eq:7a}) is arbitrary, Eq. (\ref{eq:8a}) means that the two CP maps $\mathcal{E}_{1S}$ and $\mathcal{E}_{2S}$ are the same.

Next, we write the operator sum  representations of $\mathcal{E}_{1S}$ and $\mathcal{E}_{2S}$ as
\begin{equation}
\label{eq:9a}
\begin{aligned}
\mathcal{E}_{1S}(*)=\sum_{i}E_{i}\, * \,E_{i}^{\dagger},\ \ \ \sum_{i}E_{i}^{\dagger}E_{i}=I_{S}, \\
\mathcal{E}_{2S}(*)=\sum_{j}F_{j}\, * \,F_{j}^{\dagger},\ \ \ \sum_{j}F_{j}^{\dagger}F_{j}=I_{S},
\end{aligned}
\end{equation}
where $I_S$ is the identity operator on $\cH_S$, $E_i= \langle i_{E}\vert U\vert \psi_{1E}\rangle$ and $F_j= \langle j_{E}\vert U\vert \psi_{2E}\rangle$ are linear operators on $\cH_S$, and $\left\lbrace \vert i_{E}\rangle \right\rbrace$ is an orthonormal basis of  $\cH_E$.
Since  $\mathcal{E}_{1S}=\mathcal{E}_{2S}$, using Theorem 8.2 of Ref. \cite{1a}, we have $E_i=\sum_j \tilde{v}_{ij} F_j$, where $\tilde{v}_{ij}$ are elements of a unitary matrix. In other words, we have
\begin{equation}
\label{eq:10a}
\begin{aligned}
\langle i_{E}\vert U\vert \psi_{1E}\rangle=\sum_j \tilde{v}_{ij} \langle j_{E}\vert U\vert \psi_{2E}\rangle \ \\
 = \langle \tilde{i}_{E}\vert U\vert \psi_{2E}\rangle \qquad\quad \\
= \langle i_{E}\vert  V_E  U\vert \psi_{2E}\rangle, \quad  \
\end{aligned}
\end{equation}
where $\vert \tilde{i}_{E}\rangle = \sum_j \tilde{v}_{ij}^* \vert j_{E}\rangle =  V_E^\dagger \vert i_{E}\rangle  $ and $V_E$ is a unitary operator on  $\cH_E$. Note that $V_E$ depends on the operators $E_i$ and $F_j$. In other words, it depends on the kets $\vert \psi_{1E}\rangle$ and  $\vert \psi_{2E}\rangle$.

Then, we decompose the unitary operator on the whole system-environment $U$ as
\begin{equation}
\label{eq:11a}
U=\sum_l G_l \otimes B_l, 
\end{equation}
where $\left\lbrace G_l \right\rbrace $ is an orthonormal basis of  $\mathcal{L}(\cH_S)$, with respect to
the Hilbert–Schmidt inner product, and $B_l$ are linear operators in  $\mathcal{L}(\cH_E)$. Inserting Eq. (\ref{eq:11a}) into Eq. (\ref{eq:10a}), we have
\begin{equation}
\label{eq:12}
\begin{aligned}
\sum_l G_l  \ \langle i_{E}\vert B_l\vert \psi_{1E}\rangle=\sum_l G_l  \  \langle i_{E}\vert V_E B_l\vert \psi_{2E}\rangle.
\end{aligned}
\end{equation}
Since $G_l$ are linearly independent, and $\left\lbrace \vert i_{E}\rangle \right\rbrace$ is an orthonormal basis of  $\cH_E$, we conclude that (for all $l$)
\begin{equation}
\label{eq:13}
B_l\vert \psi_{1E}\rangle = V_E B_l\vert \psi_{2E}\rangle.
\end{equation}
Therefore, for all $l$ and $m$, we have 
\begin{equation}
\label{eq:14}
\langle \psi_{1E}\vert B_m^\dagger B_l\vert \psi_{1E}\rangle =\langle \psi_{2E}\vert B_m^\dagger B_l\vert \psi_{2E}\rangle.
\end{equation}
Note that, in Eq. (\ref{eq:7a}) and so in  Eq. (\ref{eq:14}), we can arbitrarily change  $\vert \psi_{2E}\rangle$   for a fixed $\vert \psi_{1E}\rangle$. This means that
\begin{equation}
\label{eq:15}
\langle \psi_{E}\vert B_m^\dagger B_l\vert \psi_{E}\rangle =b_{ml},
\end{equation}
where $\vert \psi_{E}\rangle$ is any arbitrary (normalized) state in $\cH_E$ and $b_{ml}$ is some fixed complex scalar independent of  $\vert \psi_{E}\rangle$. (Obviously, $b_{ll}$ are positive.)

We first choose $l=m$ in  Eq. (\ref{eq:15}). So, for the positive operator $B_l^\dagger B_l$, choosing $\vert \psi_{E}\rangle$ to be its eigenvectors,  we conclude that all the eigenvalues of $B_l^\dagger B_l$ are the same and equal to  $b_{ll}$. In other words, $B_l^\dagger B_l=b_{ll} I_E$, where $I_E$ is the identity operator on $\cH_E$. Therefore, $B_l =\sqrt{b_{ll}} U_l$ with some unitary operator $U_l$ on $\cH_E$.

Now,  Eq. (\ref{eq:15}), for $l \neq m$,  reads
\begin{equation}
\label{eq:16}
\langle \psi_{E}\vert U_m^\dagger U_l\vert \psi_{E}\rangle = \frac{b_{ml}}{\sqrt{b_{ll}b_{mm}}}=\tilde{b}_{ml}.
\end{equation}
Note that $U_m^\dagger U_l$ is a unitary operator. Again, choosing $\vert \psi_{E}\rangle$ to be its eigenvectors,  we conclude that all the eigenvalues of $U_m^\dagger U_l$ are the same and equal to the fixed  phase factor $\tilde{b}_{ml}$.
Therefore,  $U_m^\dagger U_l=\tilde{b}_{ml} I_E$. Consequently, up to a phase factor, all $ U_l$ are the same as some fixed unitary operator $U_E$.

In summary, all $B_l$ in Eq. (\ref{eq:11a}), up to a phase factor, are equal to  $\sqrt{b_{ll}}U_E$, where $U_E$ is some fixed unitary operator. Inserting these  $B_l$ in Eq. (\ref{eq:11a}), we conclude that $U$ is factorized as  $U=U_S\otimes U_E$. $\qquad\qquad\qquad\qquad\qquad\blacksquare$

The above Lemma is in fact the same as Lemma 2 of Ref. \cite{7a}. So, our proof is another (maybe simpler) proof for it. 
In addition, a special case of Lemma \ref{lem1}, i.e., when $U=e^{-iHt}$ with the time-independent Hamiltonian $H$ in  Eq. (\ref{eq:2a}), has been proved in Ref. \cite{34a}.

Note that Lemma \ref{lem1} means that, for a  non-factorisable $U$, there exists at least one linear operator   
 $Y \in \mathrm{ker} \ \mathrm{Tr}_E$ such that  $\mathrm{Tr}_{E} \circ \mathrm{Ad}_U (Y) \neq 0$ 
But, Lemma \ref{lem1} does not guarantee that this $Y$ is also a valid $R$ in Eq. (\ref{eq:4a}).
The Hermitian operator $R$ in Eq. (\ref{eq:4a}) is such that adding it to $\rho_{S}\otimes\tilde{\omega}_E$ results in a valid density operator $\rho_{SE}$.

The authors of Ref. \cite{34a} tried to justify that some of $Y \in \mathrm{ker} \ \mathrm{Tr}_E$ for which $\mathrm{Tr}_{E} \circ \mathrm{Ad}_U (Y) \neq 0$  are valid operators $R$ in Eq. (\ref{eq:4a}).
But, we think that a   rigorous proof is needed which we will give in the next section.

\section{Deviation from completely positive reduced dynamics}\label{sec:C}

Consider the set $\mathcal{B}_S=\left\lbrace \rho_S^{(l)} \right\rbrace $, including ${(d_S)}^2$ linearly independent states $ \rho_S^{(l)} \in \mathcal{D}_S$. ($d_S$ is the dimension of $\cH_S$.)
So, each $A_S \in \mathcal{L}(\cH_S)$ can be expanded as 
\begin{equation}
\label{eq:17}
\begin{aligned}
A_{S}=\sum_{l=1}^{{(d_S)}^2} c_l \rho_{S}^{(l)},
\end{aligned}
\end{equation}  
   where $c_l$ are unique complex coefficients. In particular,  each $ \rho_S  \in \mathcal{D}_S$ can be expanded as 
\begin{equation}
\label{eq:18}
\begin{aligned}
\rho_{S}=\sum_{l=1}^{{(d_S)}^2} a_l \rho_{S}^{(l)},
\end{aligned}
\end{equation}  
  with    real coefficients  $a_l$.

Linear independence of  $\rho_{S}^{(l)} \in \mathcal{B}_S$ results in the linear independence of  $\rho_{S}^{(l)}\otimes \tilde{\omega}_E$, where  $ \tilde{\omega}_E $ is a fixed  state of the environment and $1 \leq l \leq {(d_S)}^2$. We can add $ {(d_S)}^2{(d_E)}^2 -{(d_S)}^2 $ other linearly independent states $\rho_{SE}^{(l)} \in  \mathcal{D}$, where ${(d_S)}^2 +1 \leq l \leq {(d_S)}^2{(d_E)}^2 $ and $d_E$ is the dimension of $\cH_E$, to construct the set $\mathcal{B}=\left\lbrace  \rho_{S}^{(1)}\otimes \tilde{\omega}_E, \dots , \rho_{S}^{({(d_S)}^2 )}\otimes \tilde{\omega}_E, \rho_{SE}^{({(d_S)}^2 +1)}, \dots , \rho_{SE}^{({(d_S)}^2{(d_E)}^2)} \right\rbrace $ as a basis of $\mathcal{L}(\cH_S \otimes \cH_E)$.

Note that, from  Eq. (\ref{eq:18}), for each $\rho_{SE}^{(i)} \in \mathcal{B} $ and $i  > {(d_S)}^2$, we have
\begin{equation}
\label{eq:19}
\begin{aligned}
\rho_{S}^{(i)}=\mathrm{Tr}_{E}(\rho_{SE}^{(i)}) =\sum_{l=1}^{{(d_S)}^2} a_l^{(i)} \rho_{S}^{(l)}.
\end{aligned}
\end{equation}
Therefore, from Eq. (\ref{eq:4a}), we have
\begin{equation}
\label{eq:20}
\begin{aligned}
\rho_{SE}^{(i)}=\rho^{(i)}_{S}\otimes\tilde{\omega}_E + R^{(i)}=\sum_{l=1}^{{(d_S)}^2} a_l^{(i)} \rho_{S}^{(l)}\otimes \tilde{\omega}_E + R^{(i)},
\end{aligned}
\end{equation}
where  $R^{(i)}$ is a Hermitian operator in  $\mathcal{L}(\cH_S \otimes \cH_E)$ such that $\mathrm{Tr}_{E} (R^{(i)})=0$.

The set $\mathcal{B} $ is a basis of $\mathcal{L}(\cH_S \otimes \cH_E)$. So, we can expand each $Y \in \mathrm{ker} \ \mathrm{Tr}_E$ as
\begin{equation}
\label{eq:21}
\begin{aligned}
Y=\sum_{l=1}^{{(d_S)}^2} d_l  \rho_{S}^{(l)}\otimes \tilde{\omega}_E +\sum_{l={(d_S)}^2 +1}^{{(d_S)}^2{(d_E)}^2} d_l\rho_{SE}^{(l)},
\end{aligned}
\end{equation}
 with   complex coefficients  $d_l$. Inserting   Eq. (\ref{eq:20}) (that is  for $ \rho_{SE}^{(l)}$  with $l  > {(d_S)}^2$) into  Eq. (\ref{eq:21}), we have
\begin{equation}
\label{eq:22}
\begin{aligned}
Y=\sum_{l=1}^{{(d_S)}^2} d_l  \rho_{S}^{(l)}\otimes \tilde{\omega}_E  \qquad\qquad\qquad\qquad\qquad\quad   \\
+\sum_{l={(d_S)}^2 +1}^{{(d_S)}^2{(d_E)}^2} d_l (\sum_{i=1}^{{(d_S)}^2} a_i^{(l)} \rho_{S}^{(i)}\otimes \tilde{\omega}_E + R^{(l)})   \\
=\sum_{l=1}^{{(d_S)}^2} \tilde{d}_l  \rho_{S}^{(l)}\otimes \tilde{\omega}_E + \sum_{l={(d_S)}^2 +1}^{{(d_S)}^2{(d_E)}^2} d_l  R^{(l)}. \qquad
\end{aligned}
\end{equation}
Tracing over the environment results that
\begin{equation}
\label{eq:23}
\begin{aligned}
0= \mathrm{Tr}_{E} (Y)=
 \sum_{l=1}^{{(d_S)}^2} \tilde{d}_l  \rho_{S}^{(l)}.
\end{aligned}
\end{equation}
Since all $\rho_{S}^{(l)} \in \mathcal{B}_S$ are linearly independent, all the coefficients $ \tilde{d}_l$ are zero. So, from Eq. (\ref{eq:22}), we conclude that
\begin{equation}
\label{eq:24}
\begin{aligned}
Y= \sum_{l={(d_S)}^2 +1}^{{(d_S)}^2{(d_E)}^2} d_l  R^{(l)}, 
\end{aligned}
\end{equation}
 with  the coefficients  $d_l$ in  Eq. (\ref{eq:21}).

Consequently, we have
\begin{equation}
\label{eq:25}
\begin{aligned}
Y_S^\prime=\mathrm{Tr}_{E} \circ\mathrm{Ad}_U (Y)= \sum_{l={(d_S)}^2 +1}^{{(d_S)}^2{(d_E)}^2} d_l  R_S^{\prime(l)}, 
\end{aligned}
\end{equation}
where $R_S^{\prime(l)}=\mathrm{Tr}_{E} \circ\mathrm{Ad}_U ( R^{(l)})$. Obviously, if $Y_S^\prime \neq 0$, then at least one of $R_S^{\prime(l)}$ must be nonzero.

Now, consider one $\rho_{SE}^{(i)}$ in Eq. (\ref{eq:20}) for which $R_S^{\prime(i)} \neq 0$. Using  Eqs. (\ref{eq:1a}) and  (\ref{eq:5a}), we have
\begin{equation}
\label{eq:26}
\begin{aligned}
\rho_{S}^{\prime(i)}=\mathrm{Tr}_{E} \circ\mathrm{Ad}_U (\rho_{SE}^{(i)})=\mathcal{E}_S (\rho_{S}^{(i)})+ R_S^{\prime(i)}.
\end{aligned}
\end{equation}
The nonzero term $R_S^{\prime(i)}$ is the deviation from CP reduced dynamics $\mathcal{E}_S (*) =\mathrm{Tr}_{E} \circ \mathrm{Ad}_U (* \otimes \tilde{\omega}_E)$.

In summary, we have proved the following Proposition.
\begin{prop}
\label{pro:1}
Consider a non-factorisable unitary evolution of the whole system-environment $U$. According to Lemma \ref{lem1}, there exists at least one  $Y \in \mathrm{ker} \ \mathrm{Tr}_E$ such that $Y_S^\prime=\mathrm{Tr}_{E} \circ\mathrm{Ad}_U (Y) \neq 0$. 
 So, from Eqs. (\ref{eq:25}) and  (\ref{eq:26}), we conclude that there exists at least one  $\rho_{SE}^{(i)} \in \mathcal{B}$ such that its reduced dynamics in Eq.  (\ref{eq:26}) deviates from the CP map $\mathcal{E}_S (*) =\mathrm{Tr}_{E} \circ \mathrm{Ad}_U (* \otimes \tilde{\omega}_E)$. 
\end{prop}

We end this section with a note on non-Markovianity. 
 The trace distance between two states $\rho$ and $\sigma$, as  a measure of the distinguishability   of these two states,
  is defined as
 \begin{equation}
\label{eq:27a}
 D(\rho , \sigma)= \frac{1}{2}\mathrm{Tr}(\vert\rho - \sigma\vert),
\end{equation}
   where $\vert Z \vert = \sqrt{Z^{\dagger}Z}$ \cite{1a}.    
After a joint system-environment time evolution, if the trace distance  between two system reduced states $\rho_S$ and $\sigma_S$ decreases, this can be interpreted as a flow of information from the system to the environment. Conversely, if  the trace distance  between $\rho_S$ and $\sigma_S$ increases, this can be interpreted as a flow of information from the environment to the system \cite{3, 4}.

The above interpretation for the trace distance variations leads to a proposed definition of the non-Markovianity of a reduced dynamics:  If during a   time evolution, we first encounter    a flow of information from the system to the environment, and then backflow of information from the environment to the system, then the reduced dynamics is called non-Markovian \cite{3, 4}. In other words, the backflow of information is the sign of the non-Markovianity of the process.

When the system-environment evolution is factorisable as $U_S \otimes U_E$, then the reduced dynamics of the system in Eq. (\ref{eq:1a}) is also unitary, and so
 \begin{equation}
\label{eq:28a}
 D(\rho^\prime_S , \sigma^\prime_S)=  D(\rho_S , \sigma_S),
\end{equation}
for any arbitrary system-environment initial states $\rho_{SE}$ and  $\sigma_{SE}$. We can interpret it as the preservation of the information within the system. There is no information flow between the system and the environment. Of course, this is an expected result, since there is no interaction  between the system and the environment.

But, when the system and the environment interact, and so the
 system-environment unitary evolution is non-factorisable, we expect information flow between the system and the environment.
 
Assume that the  system-environment  evolution, from $t=0$ to some $t_1 >0$, is given by a non-factorisable unitary operator $U$. Consider an initial state of the system-environment as  $\rho_{SE}^{(i)} \in \mathcal{B}$ for which $R_S^{\prime(i)}$ in Eq.  (\ref{eq:26}) is nonzero. Also, consider another system-environment initial state $\sigma_{SE}=\rho^{(i)}_{S}\otimes\tilde{\omega}_E$.
Therefore, at $t=0$, we have $\sigma_{S}(0)=\rho_{S}^{(i)}(0)=\rho^{(i)}_{S}$. But, at $t=t_1$, we have   $\sigma_{S}(t_1)=\mathcal{E}_S (\rho_{S}^{(i)})$ and $\rho_{S}^{(i)}(t_1)=\rho_{S}^{\prime(i)}=\mathcal{E}_S (\rho_{S}^{(i)})+ R_S^{\prime(i)}$. Consequently, 
\begin{equation}
\label{eq:29a}
\begin{aligned}
 D(\rho_{S}^{(i)}(0) , \sigma_S(0))=0,  \qquad\qquad\qquad\quad  \\
 D(\rho_{S}^{(i)}(t_1) , \sigma_S(t_1))=\frac{1}{2}\mathrm{Tr}(\vert R_S^{\prime(i)} \vert)>0,
\end{aligned}
\end{equation}
i.e., the distance between the two reduced states increases after the time evolution. We can interpret this result as the flow of  information from the environment to the system.

Conversely, assume that the  system-environment  evolution, from some $t_2 <0$ to $t=0$, is given by $U^\dagger$.
Consider  two system-environment initial states as $\rho_{SE}^{(i)}(t_2)= U \rho_{SE}^{(i)} U^{\dagger}$  and $ \sigma_{SE}(t_2)=U \rho_{S}^{(i)}\otimes\tilde{\omega}_E U^{\dagger}$.
Therefore, using Eq.  (\ref{eq:26}), the difference between the two system initial states is
 \begin{equation}
\label{eq:30a}
\rho_{S}^{(i)}(t_2)-\sigma_{S}(t_2)=\mathrm{Tr}_{E}(\rho_{SE}^{(i)}(t_2)-\sigma_{SE}(t_2))=  R_S^{\prime(i)},
\end{equation}
   while the system final states are the same: 
\begin{equation}
\label{eq:31a}
\begin{aligned}
\rho_{S}^{(i)}(0)=\mathrm{Tr}_{E} \circ \mathrm{Ad}_{U^\dagger} (\rho_{SE}^{(i)}(t_2))=\rho_{S}^{(i)}, \\
\sigma_{S}(0)=\mathrm{Tr}_{E} \circ \mathrm{Ad}_{U^\dagger}(\sigma_{SE}(t_2))=\rho_{S}^{(i)},  \ 
\end{aligned}
\end{equation}
since, for $R^{(i)}$ in Eq. (\ref{eq:20}), we have  $\mathrm{Tr}_{E} (R^{(i)})=0$.
Consequently, under the non-factorisable system-environment unitary evolution $U^\dagger$, the trace distance between the two above  system  states  decreases, which can be interpreted as the flow of  information from the  system to the environment.

In summary, we have shown that any non-factorisable unitary evolution $U$ includes information exchange between the system and the environment. So, we may guess that any non-factorisable $U$ can exhibit non-Markovian features. However, finding whether   any arbitrary non-factorisable $U$ includes \textit{backflow} of information from the environment to the system needs a more  rigorous study.

\section{Deviation from linear reduced dynamics}\label{sec:D}

Yet, we have focused on  deviation from CP reduced dynamics. Proposition \ref{pro:1} can be generalized readily to include all linear (Hermitian) reduced dynamics too.

Again, consider the set $\mathcal{B}_S=\left\lbrace \rho_S^{(l)} \right\rbrace $, including ${(d_S)}^2$ linearly independent states $ \rho_S^{(l)} \in \mathcal{D}_S$. Next, consider the set $\hat{\mathcal{B}} =\left\lbrace \rho_{SE}^{(l)} \right\rbrace $, including ${(d_S)}^2$ system-environment states $ \rho_{SE}^{(l)}$ such that $\rho_S^{(l)} = \mathrm{Tr}_{E} (\rho_{SE}^{(l)})$. Therefore, the  states $ \rho_{SE}^{(l)}$ are also linearly independent, otherwise, the  states $ \rho_{S }^{(l)}$ wouldn't be so. 
Note that, in this section, we do not restrict ourselves to the case that, for $1 \leq l \leq {(d_S)}^2$,   $ \rho_{SE}^{(l)}=\rho_S^{(l)}\otimes \tilde{\omega}_E$. The only restriction is that $\rho_S^{(l)} = \mathrm{Tr}_{E} (\rho_{SE}^{(l)})$,  otherwise, $ \rho_{SE}^{(l)}$ are arbitrary.

From Eq. (\ref{eq:17}), we know that each linear operator $A_S \in \mathcal{L}(\cH_S)$ can be expanded using  $\mathcal{B}_S$. Now, we define the linear trace-preserving \textit{assignment map} $\Lambda_S$ as
\begin{equation}
\label{eq:27}
\begin{aligned}
\Lambda_S(A_{S})=\sum_{l=1}^{{(d_S)}^2} c_l \Lambda_S(\rho_{S}^{(l)})=\sum_{l=1}^{{(d_S)}^2} c_l \rho_{SE}^{(l)}.
\end{aligned}
\end{equation} 
The  assignment map $\Lambda_S$ is Hermitian by construction: When $A_{S}$ is a Hermitian operator,    the coefficients  $c_l$ in  Eq. (\ref{eq:17}) are real, and so $\Lambda_S(A_{S})$ in the above equation is also a Hermitian operator.

We can define the subspace $\mathcal{V}  \subset \mathcal{L}(\cH_S \otimes \cH_E)$ as the subspace spanned by the states $ \rho_{SE}^{(l)} \in \hat{\mathcal{B}} $, i.e., for each $X\in \mathcal{V}$, we have
\begin{equation}
\label{eq:28}
\begin{aligned}
X=\sum_{l=1}^{{(d_S)}^2} e_l \rho_{SE}^{(l)},
\end{aligned}
\end{equation}
 where $e_l$ are unique complex coefficients. From Eq. (\ref{eq:27}), we see that $\Lambda_S(A_{S})\in \mathcal{V}$. In other words,   $\Lambda_S$ maps $ \mathcal{L}(\cH_S)$ to $\mathcal{V}$.

Consider the set of possible  initial states of the system-environment as $\mathcal{S}=\mathcal{V}\cap\mathcal{D}$. So, for each initial state  $ \rho_{SE}  \in\mathcal{S}$, we have
\begin{equation}
\label{eq:29}
\begin{aligned}
\rho_{SE}=\sum_{l=1}^{{(d_S)}^2} \hat{a}_l \rho_{SE}^{(l)}, \qquad\qquad\quad \\
\rho_{S}= \mathrm{Tr}_{E} (\rho_{SE})=\sum_{l=1}^{{(d_S)}^2} \hat{a}_l \rho_{S}^{(l)},
\end{aligned}
\end{equation}
 with    real coefficients  $\hat{a}_l$. 
  Now, using Eqs. (\ref{eq:1a}), (\ref{eq:27}) and (\ref{eq:29}), we have
\begin{equation}
\label{eq:30}
\rho_{S}^{\prime}=\mathrm{Tr}_{E} \circ \mathrm{Ad}_U (\rho_{SE})=\mathrm{Tr}_{E} \circ \mathrm{Ad}_U \circ\Lambda_S(\rho_{S}) \equiv \Phi_S(\rho_{S}).
\end{equation}
Since $\mathrm{Tr}_{E}$ and $\mathrm{Ad}_U$ are CP maps \cite{1a}, and  $\Lambda_S$ was constructed as a Hermitian map in Eq. (\ref{eq:27}), the reduced dynamics $\Phi_S=\mathrm{Tr}_{E} \circ \mathrm{Ad}_U \circ\Lambda_S$ is a (linear trace-preserving) Hermitian map, in general.

Note that $\Phi_S$ in Eq. (\ref{eq:30}) gives the reduced dynamics correctly only for the system-environment initial states  $\rho_{SE}\in\mathcal{S}=\mathcal{V}\cap\mathcal{D}$. In other words,  $\Phi_S$ gives the reduced dynamics correctly only for the system  initial states $\rho_{S}\in\mathcal{S}_S=\mathrm{Tr}_{E} \ \mathcal{S} = \mathrm{Tr}_{E} (\mathcal{V}\cap\mathcal{D})$, and  in general  $\mathrm{Tr}_{E} (\mathcal{V}\cap\mathcal{D}) \neq \mathcal{D}_S$.
Only when $\Lambda_S$ is a positive map on the whole $ \mathcal{L}(\cH_S)$, we have $\mathcal{S}_S=\mathrm{Tr}_{E} (\mathcal{V}\cap\mathcal{D})=\mathcal{D}_S$. (We will come back to this point in the next section.)
It can be shown that when  $\Lambda_S$ in Eq. (\ref{eq:27}) is   positive, then it is CP too, and $\Lambda_S(*)=(*)\otimes \tilde{\omega}_E$, where $ \tilde{\omega}_E $ is some fixed state of the environment \cite{5a, 35, 36}.
 Consequently, the reduced dynamics in Eq. (\ref{eq:30}) is also CP as $\mathcal{E}_S (*) =\mathrm{Tr}_{E} \circ \mathrm{Ad}_U (* \otimes \tilde{\omega}_E)$  which is the case studied in the previous section.

 We can add $ {(d_S)}^2{(d_E)}^2 -{(d_S)}^2 $ other linearly independent states $\rho_{SE}^{(l)} \in  \mathcal{D}$ to the set $\hat{\mathcal{B}}$ to construct the set $\tilde{\mathcal{B}}$,  including $ {(d_S)}^2{(d_E)}^2$ linearly independent states, as a basis of $\mathcal{L}(\cH_S \otimes \cH_E)$. 
 For each $\rho_{SE}^{(i)} \in \tilde{\mathcal{B}} $ and $i  > {(d_S)}^2$, since $\mathcal{B}_S$ is a basis of $\mathcal{L}(\cH_S)$, we have
\begin{equation}
\label{eq:31}
\begin{aligned}
\rho_{S}^{(i)}=\mathrm{Tr}_{E}(\rho_{SE}^{(i)}) =\sum_{l=1}^{{(d_S)}^2} a_l^{(i)} \rho_{S}^{(l)},
\end{aligned}
\end{equation}
 with    real coefficients  $a_l^{(i)}$. Then, using the assignment map $\Lambda_S$ in Eq. (\ref{eq:27}), we have
\begin{equation}
\label{eq:32}
\begin{aligned}
\rho_{SE}^{(i)}=\Lambda_S(\rho_{S}^{(i)}) + R^{(i)}=\sum_{l=1}^{{(d_S)}^2} a_l^{(i)} \rho_{SE}^{(l)}  + R^{(i)},
\end{aligned}
\end{equation}
where  $R^{(i)}$ is a Hermitian operator in  $\mathcal{L}(\cH_S \otimes \cH_E)$ such that $\mathrm{Tr}_{E} (R^{(i)})=0$. Note that, since $\rho_{SE}^{(i)}$ is linearly independent from $\rho_{SE}^{(l)}$ with $1 \leq l \leq {(d_S)}^2$, the operator  $R^{(i)}$ is nonzero.

When the initial state of the system-environment is $\rho_{SE}^{(i)}$  with $i  > {(d_S)}^2$, the reduced dynamics of the system, using Eqs. (\ref{eq:1a}), (\ref{eq:30}) and (\ref{eq:32}), is as
\begin{equation}
\label{eq:33}
\rho_{S}^{\prime (i)}=\mathrm{Tr}_{E} \circ \mathrm{Ad}_U (\rho_{SE}^{(i)})=\Phi_S(\rho_{S}^{(i)}) + R_S^{\prime(i)}, 
\end{equation}
where $R_S^{\prime(i)}=\mathrm{Tr}_{E} \circ\mathrm{Ad}_U ( R^{(i)})$.
Therefore, if,  for a given system-environment unitary evolution $U$, we have  $R_S^{\prime(i)}=0$, then, at least for this $U$, we can add $\rho_{SE}^{(i)}$ to the set $\mathcal{S}$ for which the reduced dynamics is given by $\Phi_S$. In fact, we can add $\rho_{SE}^{(i)}$ to the linearly independent set  $\hat{\mathcal{B}}$ and extend the subspace $\mathcal{V}$, and so the set $\mathcal{S}=\mathcal{V}\cap\mathcal{D}$.

Can we follow a similar  procedure for all $\rho_{SE}^{(i)}\in \tilde{\mathcal{B}}$  with $i  > {(d_S)}^2$ ?
In other words, for a given $U$, whether  $R_S^{\prime(i)}=0$ for all $i  > {(d_S)}^2$. We can follow a similar procedure as given from Eq. (\ref{eq:21}) onward, to show that this is not the case for any non-factorisable $U$.

Since the set $\tilde{\mathcal{B}}$ is a basis of $\mathcal{L}(\cH_S \otimes \cH_E)$, we can expand each $Y \in \mathrm{ker} \ \mathrm{Tr}_E$ as
\begin{equation}
\label{eq:34}
\begin{aligned}
Y= \sum_{l=  1}^{{(d_S)}^2 {(d_E)}^2} d_l \rho_{SE}^{(l)},
\end{aligned}
\end{equation}
 with   complex coefficients  $d_l$ and $\rho_{SE}^{(l)}\in \tilde{\mathcal{B}}$. Inserting   Eq. (\ref{eq:32}) (that is  for $ \rho_{SE}^{(l)}$ with $l  > {(d_S)}^2$) into  Eq. (\ref{eq:34}), we have
\begin{equation}
\label{eq:35}
\begin{aligned}
Y=\sum_{l=1}^{{(d_S)}^2} \tilde{d}_l  \rho_{SE}^{(l)} + \sum_{l={(d_S)}^2 +1}^{{(d_S)}^2{(d_E)}^2} d_l  R^{(l)}. 
\end{aligned}
\end{equation}
Again, from $\ \mathrm{Tr}_E (Y)=0$, we conclude that all the coefficients $\tilde{d}_l$ are zero. Therefore
\begin{equation}
\label{eq:36}
\begin{aligned}
Y= \sum_{l={(d_S)}^2 +1}^{{(d_S)}^2{(d_E)}^2} d_l  R^{(l)}, \qquad\qquad\qquad\qquad \\
Y_S^\prime=\mathrm{Tr}_{E} \circ\mathrm{Ad}_U (Y)= \sum_{l={(d_S)}^2 +1}^{{(d_S)}^2{(d_E)}^2} d_l  R_S^{\prime(l)}. 
\end{aligned}
\end{equation}
From Lemma \ref{lem1}, we know that, for a non-factorisable $U$, there exists at least one $Y \in \mathrm{ker} \ \mathrm{Tr}_E$ such that $Y_S^\prime \neq 0$. Therefore, at least one $R_S^{\prime(l)}$ in the above equation is nonzero.
So, the reduced dynamics in Eq. (\ref{eq:33}) deviates from the linear map $\Phi_S$, at least for one $\rho_{SE}^{(i)}\in \tilde{\mathcal{B}}$ with $i  > {(d_S)}^2$.

In summary, we have generalized the Proposition \ref{pro:1} as the following.
\begin{prop}
\label{pro:2}
Consider a non-factorisable unitary evolution of the whole system-environment $U$. According to Lemma \ref{lem1}, there exists at least one  $Y \in \mathrm{ker} \ \mathrm{Tr}_E$ such that $Y_S^\prime=\mathrm{Tr}_{E} \circ\mathrm{Ad}_U (Y) \neq 0$. 
 So, from Eqs. (\ref{eq:33}) and  (\ref{eq:36}), we conclude that there exists at least one  $\rho_{SE}^{(i)} \in \tilde{\mathcal{B}}$ such that its reduced dynamics in Eq.  (\ref{eq:33}) deviates from the (linear trace-preserving)  Hermitian map $\Phi_S=\mathrm{Tr}_{E} \circ \mathrm{Ad}_U \circ\Lambda_S$.
\end{prop}

At the first look, it seems that
  when $\Phi_S$ is a non-positive map, the above Proposition gives us rather a trivial result: A non-positive  $\Phi_S$ maps some  $\tilde{\rho}_{S}\in\mathcal{D}_S$ to non-positive operators, and so cannot describe the reduced dynamics correctly for these initial states of the system $\tilde{\rho}_{S}$. In other words, it cannot give correctly the reduced dynamics for initial states of the system-environment $\tilde{\rho}_{SE}$, for which we have $\mathrm{Tr}_{E}(\tilde{\rho}_{SE})=\tilde{\rho}_{S}$. 

But, obviously,  when $\Phi_S$ is a positive map, then Proposition \ref{pro:2} provides us a nontrivial generalization of  Proposition \ref{pro:1}.

Let's end this section with a simple illustrating example. Consider the case that the joint system-environment unitary evolution is given by the swap operator $U_{swap}$. It acts as  $U_{swap} \vert \psi\rangle \vert \phi\rangle=\vert \phi\rangle \vert \psi\rangle$, and is non-factorisable, since  $U_{swap} \neq \mathrm{Tr}_{E}(U_{swap})\otimes  \mathrm{Tr}_{S}(U_{swap})=I_S\otimes I_E$.

In Ref. \cite{11a}, $U_{swap}$ has been used to give a general method for physical implementing any (linear trace-preserving) Hermitian map $\Phi_S$ as the reduced dynamics of the system: Consider the set $\mathcal{S}_S= \left\lbrace \rho_S \right\rbrace$ of all system states $\rho_S$ for which $\Phi_S(\rho_S)$ are valid states (density operators).
Construct the set of possible initial states of the system-environment   as (the convex hull of) the set  $\mathcal{S} =\left\lbrace \rho_{SE}= \rho_S\otimes \Phi_{S}(\rho_S) \mid \rho_S \in \mathcal{S}_S \right\rbrace$.
 Now, for each $\rho_{SE} \in \mathcal{S}$,  we have  $\mathrm{Tr}_{E} \circ \mathrm{Ad}_{U_{swap}} (\rho_{SE})= \Phi_{S}(\rho_S)$.

The Hermitian map $\Phi_S$ gives correctly the reduced dynamics for each  system-environment  initial state $\rho_{SE} \in \mathcal{S}$. But, if we consider a  system-environment initial state as, e.g., $\sigma_{SE}= \rho_S \otimes \sigma_{E}$, where  $\rho_{S} \in \mathcal{S}_S$ but $ \sigma_{E} \neq \Phi_{S}(\rho_S)$, then the reduced dynamics of the system, for the non-factorisable unitary evolution $U_{swap}$, is given by  $\mathrm{Tr}_{E} \circ \mathrm{Ad}_{U_{swap}} (\sigma_{SE})= \sigma_E$ that differs from   $\Phi_{S}(\rho_S)$.
Therefore, for the system-environment initial state $\sigma_{SE}$, the the result of the linear map $ \Phi_{S}$ on the system initial state $\rho_S=\mathrm{Tr}_{E} (\sigma_{SE})$ does not give correctly the reduced dynamics of the system.

The above example clarifies that when in  Eq.  (\ref{eq:33}) $\Phi_S(\rho_{S}^{(i)})$  is a valid density operator, but    $R_S^{\prime(i)} \neq 0$, then Proposition \ref{pro:2} gives us a nontrivial result, even for a non-positive map $\Phi_S$.

\section{Positivity for one $U$ can imply non-linearity for another $U$}\label{sec:E}

In the previous section, we have started with considering the sets $\mathcal{B}_S=\left\lbrace \rho_S^{(l)} \right\rbrace $ and $\hat{\mathcal{B}} =\left\lbrace \rho_{SE}^{(l)} \right\rbrace $ both  including ${(d_S)}^2$ linearly independent states such that $\rho_S^{(l)} = \mathrm{Tr}_{E} (\rho_{SE}^{(l)})$.
Then, we have defined the subspaces
\begin{equation}
\label{eq:37}
\begin{aligned}
\mathcal{V}= \mathrm{Span}_{\mathbb{C}} \  \hat{\mathcal{B}}  \subset \mathcal{L}(\cH_S \otimes \cH_E),  
\end{aligned}
\end{equation}  
and 
\begin{equation}
\label{eq:38}
\begin{aligned}
\mathcal{V}_S=\mathrm{Tr}_{E} \mathcal{V}=\mathrm{Span}_{\mathbb{C}} \  \mathcal{B}_S=\mathcal{L}(\cH_S).
\end{aligned}
\end{equation}  
There is a one-to-one correspondence  between the subspaces $\mathcal{V}$ and $\mathcal{V}_S$. Each $A_S \in \mathcal{V}_S=\mathcal{L}(\cH_S)$ in Eq. (\ref{eq:17}) is related only to one $A  \in \mathcal{V}$ as 
\begin{equation}
\label{eq:39}
\begin{aligned}
A=\Lambda_S(A_{S})=\sum_{l=1}^{{(d_S)}^2} c_l \rho_{SE}^{(l)},
\end{aligned}
\end{equation} 
where $\Lambda_S$ is the assignment map introduced in Eq. (\ref{eq:27}), and vice versa.

For the set of possible initial states of the system-environment as $\mathcal{S}=\mathcal{V}\cap\mathcal{D}$, Eq. 
 (\ref{eq:29}) holds.
 Note that, as a special case of   Eq.  (\ref{eq:39}), each $\rho_{SE} \in \mathcal{S}$ and $\rho_{S}=\mathrm{Tr}_{E}(\rho_{SE}) \in \mathcal{S}_S=\mathrm{Tr}_{E} \ \mathcal{S}$ are related to each other by the assignment map  in Eq. (\ref{eq:27}).

  On the other hand, since  $\mathcal{V}_S=\mathcal{L}(\cH_S)$, for each $\rho_S \in \mathcal{D}_S$,  Eq.  (\ref{eq:18}) holds. 
Now,  if we have $\mathcal{S}_S=\mathcal{D}_S$, this means that the assignment map  in Eq. (\ref{eq:27}) maps each state in Eq. (\ref{eq:18}) to a valid density operator. In other words, the assignment map $\Lambda_S$  in Eq. (\ref{eq:27}) is a positive map on the whole $\mathcal{L}(\cH_S)$. Consequently, $\Lambda_S$ is also CP \cite{5a, 35, 36}, and so the reduced dynamics is CP too.

In summary, if a) there exists   a one-to-one correspondence  between the subspaces $\mathcal{V}$ and $\mathcal{V}_S$ and b) $\mathcal{S}_S=\mathcal{D}_S$, then the assignment map $\Lambda_S$  in Eq. (\ref{eq:27}) is CP, and so is the reduced dynamics for arbitrary system-environment unitary evolution $U$.
This is, in fact, a restatement of the main result of Ref. \cite{37}.

Let us relax the one-to-one correspondence  between  $\mathcal{V}$ and $\mathcal{V}_S$. For a given $U$, consider  all $\rho_{SE}^{(i)} \in \tilde{\mathcal{B}}$, with $i  > {(d_S)}^2$, for which $R_S^{\prime(i)}$ in Eq.  (\ref{eq:33}) are zero. We add all such $\rho_{SE}^{(i)}$ to the set  $\hat{\mathcal{B}}$ to construct the extended set $\hat{\mathcal{B}}^{(e)}\subseteq\tilde{\mathcal{B}}$ including $n \leq {(d_S)}^2{(d_E)}^2$ of linearly independent system-environment  states. In addition, we define the extended subspace 
\begin{equation}
\label{eq:40}
\begin{aligned}
\mathcal{V}^{(e)}= \mathrm{Span}_{\mathbb{C}} \  \hat{\mathcal{B}}^{(e)}  \subseteq  \mathcal{L}(\cH_S \otimes \cH_E).  
\end{aligned}
\end{equation}
Therefore, for each $\rho_{SE} \in \mathcal{S}^{(e)}= \mathcal{V}^{(e)}\cap \mathcal{D}$, we have
\begin{equation}
\label{eq:41}
\begin{aligned}
\rho_{SE}=\sum_{l=1}^{n} \tilde{a}_l \rho_{SE}^{(l)},
\end{aligned}
\end{equation}  
  with    real coefficients  $\tilde{a}_l$.
  Inserting $\rho_{SE}^{(l)}$ with $l  > {(d_S)}^2$ from
  Eq. (\ref{eq:32}),  we can rewrite the above equation as
\begin{equation}
\label{eq:42}
\begin{aligned}
\rho_{SE}=\sum_{l=1}^{{(d_S)}^2} a_l \rho_{SE}^{(l)}+ \sum_{l={(d_S)}^2 + 1}^{n} \tilde{a}_l R^{(l)} \\
=\Lambda_S (\rho_{S}) +  \sum_{l={(d_S)}^2 + 1}^{n} \tilde{a}_l R^{(l)} \quad  \\
=\Lambda_S (\rho_{S}) + R(\rho_{SE}),  \qquad\qquad
\end{aligned}
\end{equation}
where $\rho_{S}=\mathrm{Tr}_{E} (\rho_{SE})$ is expanded as  Eq. (\ref{eq:18}).

Remember  we have assumed that,  for ${(d_S)}^2 < l \leq n $, $R_S^{\prime(l)}$ in Eq. (\ref{eq:33}) are all zero. So, using Eq. (\ref{eq:1a}), the reduced dynamics of the system is given by  
\begin{equation}
\label{eq:43}
\rho_{S}^{\prime}=\mathrm{Tr}_{E} \circ \mathrm{Ad}_U (\rho_{SE})=\mathrm{Tr}_{E} \circ \mathrm{Ad}_U  \circ\Lambda_S (\rho_{S})=\Phi_S (\rho_{S}),
\end{equation}
for all $\rho_{S} \in \mathcal{S}^{(e)}_S= \mathrm{Tr}_{E} (\mathcal{V}^{(e)}\cap \mathcal{D})$.

Now, assume that, for a given system-environment unitary evolution $U_1$, we have a) $\mathcal{S}^{(e)}_S=\mathcal{D}_S$, and b) $\Phi_{1S}$ is a positive, but not a CP, map. When $\Phi_{1S}=\mathrm{Tr}_{E} \circ \mathrm{Ad}_{U_1}  \circ\Lambda_S$ is not CP, it means that $\Lambda_S$ is not CP, since $\mathrm{Tr}_{E}$ and $\mathrm{Ad}_{U_1}$ are CP. In fact, then   $\Lambda_S$ is a non-positive map \cite{36}. Therefore, for some  $\rho_{S} \in \mathcal{S}^{(e)}_S=\mathcal{D}_S$  expanded as Eq. (\ref{eq:18}), the operator
\begin{equation}
\label{eq:44}
\begin{aligned}
\Lambda_S (\rho_{S})=\sum_{l=1}^{{(d_S)}^2} a_l \rho_{SE}^{(l)},
\end{aligned}
\end{equation}
is not a positive operator. Consequently, the system-environment state $\rho_{SE} \in \mathcal{S}^{(e)}$, for which we have $\mathrm{Tr}_{E}(\rho_{SE})=\rho_{S }$, is as 
\begin{equation}
\label{eq:45}
\begin{aligned}
\rho_{SE}=
 \Lambda_S (\rho_{S}) + R(\rho_{SE}), 
\end{aligned}
\end{equation}
where $ R(\rho_{SE})$ is surely nonzero.
For a nonzero (Hermitian traceless) operator $ R $, there exists another unitary operator $U_2$ such that $R^\prime_S=\mathrm{Tr}_{E} \circ \mathrm{Ad}_{U_2}(R) \neq 0$ \cite{21a}. So, for this initial state $\rho_{SE} \in \mathcal{S}^{(e)}$, using Eqs. (\ref{eq:1a}) and (\ref{eq:45}), the reduced dynamics of the system, for the system-environment unitary evolution $U_2$, is given by
\begin{equation}
\label{eq:46}
\rho_{S}^{\prime}=\mathrm{Tr}_{E} \circ \mathrm{Ad}_{U_2} (\rho_{SE})=\Phi_{2S} (\rho_{S})+ R^\prime_S,
\end{equation}
where $\Phi_{2S}=\mathrm{Tr}_{E} \circ \mathrm{Ad}_{U_2}\circ\Lambda_S$ is a Hermitian map, and $R^\prime_S$ gives the nonzero deviation from this linear reduced dynamics.

In summary, we have proved the following Proposition.
\begin{prop}
\label{pro:3}
Consider a set of possible initial states of the system-environment as $ \mathcal{S}^{(e)}= \mathcal{V}^{(e)}\cap \mathcal{D}$ such that $ \mathcal{S}^{(e)}_S=\mathrm{Tr}_{E} \ \mathcal{S}^{(e)}=\mathcal{D}_S$.
If, for one unitary system-environment evolution $U_1$, the reduced dynamics of the system is given by a positive, but non-CP, map $\Phi_{1S}$, then there exists another (non-factorisable) unitary system-environment evolution $U_2$ for which deviation from linear reduced dynamics occurs for some system-environment initial states $\rho_{SE} \in \mathcal{S}^{(e)}$.
\end{prop}

Noting that the above result is based on the nonzero-ness of $R$ in Eq. (\ref{eq:45}), we can restate it in the following form too.
\begin{cor}
\label{cor:1}
Consider an arbitrary subspace $\tilde{\mathcal{V}}   \subseteq \mathcal{L}(\cH_S \otimes \cH_E)$ spanned by the states. Construct the set of possible initial states of the system-environment as $\tilde{\mathcal{S}}=\tilde{\mathcal{V}} \cap \mathcal{D}$.
Now, if a) $\tilde{\mathcal{S}}_S=\mathrm{Tr}_{E} \ \tilde{\mathcal{S}}$ is the same as $ \mathcal{D}_S$, and b) the reduced dynamics of the system for some $U_1$ is given by a   positive, but non-CP, map, then there is no one-to-one correspondence  between the subspaces $\tilde{\mathcal{V}}$ and $\tilde{\mathcal{V}}_S= \mathrm{Tr}_{E} \  \tilde{\mathcal{V}}=\mathcal{L}(\cH_S) $.
\end{cor}

No similar results can be proven for non-positive reduced dynamics or the CP ones.
When the reduced dynamics is given by a non-positive map, then, obviously, $\tilde{\mathcal{S}}_S$ cannot be the same as $ \mathcal{D}_S$.
In addition, when the reduced dynamics, for a system-environment unitary evolution $U_1$, is given by a CP map $\mathcal{E}_{1S}$, and also $\tilde{\mathcal{S}}_S=\mathcal{D}_S$, both options, i.e., the one-to-one correspondence  between    $\tilde{\mathcal{V}}$ and $\tilde{\mathcal{V}}_S=\mathcal{L}(\cH_S) $ and the non-one-to-one correspondence  between    $\tilde{\mathcal{V}}$ and $\tilde{\mathcal{V}}_S$, are possible \cite{18a}.

Also  note that the assumption that $\tilde{\mathcal{S}}_S=\mathcal{D}_S$ is necessary in Proposition \ref{pro:3} and Corollary \ref{cor:1}. Otherwise, all $\rho_S \in \mathcal{D}_S$ are not needed to be related to some  $\rho_{SE} \in  \tilde{\mathcal{S}}$ such that $\mathrm{Tr}_{E}(\rho_{SE})=\rho_{S}$. Such requirement is needed only for  $\rho_S \in \tilde{\mathcal{S}}_S$.
 Therefore, for $\rho_S \notin \tilde{\mathcal{S}}_S$, the non-positivity of the operator $\Lambda_S (\rho_S)$ makes no problem. Only when  $\tilde{\mathcal{S}}_S=\mathcal{D}_S$, the non-positivity of   $\Lambda_S (\rho_S)$ for some $\rho_S \in \mathcal{D}_S$ results in the nonzero-ness of $R$ in Eq. (\ref{eq:45})  for some $\rho_{SE }\in \tilde{\mathcal{S}} $.

Using the following Lemma, we can restate Corollary \ref{cor:1} in terms of the sets of possible initial states $\tilde{\mathcal{S}}$ and $\tilde{\mathcal{S}}_S$, instead of the related subspaces $\tilde{\mathcal{V}}$ and $\tilde{\mathcal{V}}_S$.

\begin{lemma}
\label{lem2}
 Consider an arbitrary convex set of system-environment states $\tilde{\mathcal{S}}=\left\lbrace \sigma_{SE}\right\rbrace $. Construct the subspace $\tilde{\mathcal{V}}   \subseteq \mathcal{L}(\cH_S \otimes \cH_E)$ as $\tilde{\mathcal{V}}= \mathrm{Span}_{\mathbb{C}} \ \tilde{\mathcal{S}}$. Now, there is a   one-to-one correspondence  between  the subspaces  $\tilde{\mathcal{V}}$ and $\tilde{\mathcal{V}}_S= \mathrm{Tr}_{E} \  \tilde{\mathcal{V}}$, if and only if there is a  one-to-one correspondence  between  the convex sets of states $\tilde{\mathcal{S}}$ and  $\tilde{\mathcal{S}}_S=\mathrm{Tr}_{E} \ \tilde{\mathcal{S}}$.
 \end{lemma}
\textit{Proof.}
When there is a one-to-one correspondence  between $\tilde{\mathcal{V}}$ and $\tilde{\mathcal{V}}_S$, then, trivially, this will be the case for $\tilde{\mathcal{S}}$ and $\tilde{\mathcal{S}}_S$, since $\tilde{\mathcal{S}}\subset\tilde{\mathcal{V}}$ and $\tilde{\mathcal{S}}_S\subset\tilde{\mathcal{V}}_S$. So, in the following, we focus on the  inverse.

Each $X\in \tilde{\mathcal{V}}$ is expanded as 
\begin{equation}
\label{eq:47}
\begin{aligned}
X=\sum_{j}  c_j \tilde{\sigma}_{SE}^{(j)},
\end{aligned}
\end{equation}  
   where $c_j$ are  complex coefficients, and  $\tilde{\sigma}_{SE}^{(j)}\in \tilde{\mathcal{S}}$ are not necessarily linearly independent.
Expand  each $c_j$ as
\begin{equation}
\label{eq:48}
\begin{aligned}
  c_j=(c_j^{(1)} - c_j^{(2)}) +i(c_j^{(3)} - c_j^{(4)}),
\end{aligned}
\end{equation}
 where $i=\sqrt{-1}$, and $c_j^{(1)}$, $c_j^{(2)}$, $c_j^{(3)}$ and $c_j^{(4)}$ are all positive.
Inserting  Eq. (\ref{eq:48}) into Eq. (\ref{eq:47}), and noting that the set $\tilde{\mathcal{S}}$ is convex, we can rewrite Eq. (\ref{eq:47}) as
\begin{equation}
\label{eq:49}
\begin{aligned}
 X=(a^{(1)} \sigma_{SE}^{(1)} - a^{(2)} \sigma_{SE}^{(2)}) +i(a^{(3)}\sigma_{SE}^{(3)} - a^{(4)}\sigma_{SE}^{(4)}),
\end{aligned}
\end{equation}
 where, for $1\leq m \leq 4$, we have  $a^{(m)}=\sum_{j} c_j^{(m)}$ 
  and $\sigma_{SE}^{(m)} \in \tilde{\mathcal{S}}$.

Tracing over the environment, we have 
\begin{equation}
\label{eq:50}
\begin{aligned}
 X_S=\mathrm{Tr}_{E}(X)=(a^{(1)} \sigma_{S}^{(1)} - a^{(2)} \sigma_{S}^{(2)}) +i(a^{(3)}\sigma_{S}^{(3)} - a^{(4)}\sigma_{S}^{(4)}),
\end{aligned}
\end{equation}
where $ \sigma_{S}^{(m)}= \mathrm{Tr}_{E}(\sigma_{SE}^{(m)})\in \tilde{\mathcal{S}}_S$.
When $X=0$, then obviously $X_S=\mathrm{Tr}_{E}(X)=0$.

On the other hand, when  $X_S=0$, then, in Eq. (\ref{eq:50}), we have  $a^{(1)}=a^{(2)}$ and $\sigma_{S}^{(1)}=\sigma_{S}^{(2)}$. Also, we have   $a^{(3)}=a^{(4)}$ and $\sigma_{S}^{(3)}=\sigma_{S}^{(4)}$.
Since, by assumption, there is a  one-to-one correspondence  between  the convex sets   $\tilde{\mathcal{S}}$ and  $\tilde{\mathcal{S}}_S$, the relation $\sigma_{S}^{(1)}=\sigma_{S}^{(2)}$ results that $\sigma_{SE}^{(1)}=\sigma_{SE}^{(2)}$. Also, the relation $\sigma_{S}^{(3)}=\sigma_{S}^{(4)}$ results that $\sigma_{SE}^{(3)}=\sigma_{SE}^{(4)}$. Therefore, from Eq. (\ref{eq:49}), we conclude that $X=0$.

Consequently, the  one-to-one correspondence  between  the convex sets   $\tilde{\mathcal{S}}$ and  $\tilde{\mathcal{S}}_S$ results in the  one-to-one correspondence  between  the subspaces  $\tilde{\mathcal{V}}$ and $\tilde{\mathcal{V}}_S$. $\qquad\qquad \blacksquare$

Lemma \ref{lem2} clearly shows that when there   is no one-to-one correspondence  between the subspaces $\tilde{\mathcal{V}}$ and $\tilde{\mathcal{V}}_S$, then there is no such correspondence  between the convex sets   $\tilde{\mathcal{S}}$ and  $\tilde{\mathcal{S}}_S$. So, we can rewrite the Corollary \ref{cor:1} in the following form.
\begin{cor}
\label{cor:2}
Consider the convex set of possible initial states of the system-environment as $\tilde{\mathcal{S}}=\tilde{\mathcal{V}} \cap \mathcal{D}$  such that  $\tilde{\mathcal{S}}_S=\mathrm{Tr}_{E} \ \tilde{\mathcal{S}}$ is the same as $ \mathcal{D}_S$.
Now, if  the reduced dynamics of the system for some $U_1$ is given by a a positive, but non-CP, map, then there is no one-to-one correspondence  between the the   sets   $\tilde{\mathcal{S}}$ and  $\tilde{\mathcal{S}}_S$.
\end{cor}

\section{Summary}\label{sec:F}

Our main results in this paper are summarized in two Lemmas, three Propositions and two  Corollaries.

First, in Lemma \ref{lem1}, we showed that, for each non-factorisable system-environment unitary evolution $U$, there exists at least one $Y \in \mathrm{ker} \ \mathrm{Tr}_E$ such that $Y_S^\prime=\mathrm{Tr}_{E} \circ\mathrm{Ad}_U (Y) \neq 0$. This Lemma is in fact the same as Lemma 2 of Ref. \cite{7a}, and also generalizes the main result of Ref. \cite{34a}.

An operator $Y \in \mathrm{ker} \ \mathrm{Tr}_E$ for which $Y_S^\prime \neq 0$ is not necessarily a valid $R$ in Eqs. \eqref{eq:4a} or \eqref{eq:32}. In sections  ~\ref{sec:C} and ~\ref{sec:D}, we showed that how the nonzero-ness of a 
 $Y_S^\prime$ results in the nonzero-ness of some  $R^\prime_S=\mathrm{Tr}_{E} \circ \mathrm{Ad}_{U}(R) $   for some valid  $R$ in Eqs. \eqref{eq:4a} or \eqref{eq:32}.
Consequently, for any non-factorisable $U$, deviation from linear reduced dynamics always occurs, as stated in Propositions \ref{pro:1} and  \ref{pro:2}.

In addition, in Sec. ~\ref{sec:C},  we showed that any non-factorisable system-environment unitary evolution $U$ results in information exchange between the system and the environment.
So, we expect that any non-factorisable $U$ exhibits non-Markovian effects. However, (dis)proving this conjecture needs more rigorous studies.

A linear reduced dynamics of the system $\Phi_S$ can be   either CP, positive or non-positive. But, interestingly,  the assignment   map $\Lambda_S$ has only two forms: Consider the convex set $\mathcal{S}$ of possible  initial states of the system-environment. When $\mathcal{S}_S=\mathrm{Tr}_{E} \ \mathcal{S}=\mathcal{D}_S$, the assignment   map $\Lambda_S$ can be either CP or non-positive \cite{36}.

Therefore, if  for a system-environment unitary evolution $U_1$  the reduced dynamics is given by $\Phi_{1S}=\mathrm{Tr}_{E} \circ \mathrm{Ad}_{U_1}  \circ\Lambda_S$ and if $\Phi_{1S}$ is a positive,   but  non-CP, map, we conclude that $\Lambda_S$ is non-positive.
Consequently, there exists another $U_2$ for which the reduced dynamics of the system cannot be given by a linear map for all initial states $\rho_{SE} \in \mathcal{S}$, as stated in Proposition \ref{pro:3}. In addition, there is no  one-to-one correspondence  between the    sets   $\mathcal{S}$ and  $\mathcal{S}_S=\mathcal{D}_S$, as 
 stated in Corollary \ref{cor:2}. The Corollary \ref{cor:2} is in fact a restatement of  the Corollary \ref{cor:1}, using   Lemma \ref{lem2}.

%
%
%

%


\end{document}